\documentclass[letterpaper,twocolumn,10pt]{article}
\usepackage{usenix,epsfig,endnotes}
\usepackage{amsmath,amssymb}
\usepackage{times}
\usepackage{fullpage}
\usepackage[figure,vlined,linesnumbered]{algorithm2e}
\usepackage{color}
\usepackage{rotating}
\usepackage{soul}
\usepackage{refcount}
\usepackage{helvet}
\usepackage{url}
\usepackage{refcount}
\usepackage[numbers,sort]{natbib}
\usepackage{inconsolata}
\usepackage{subfigure}

\newcommand{\name}{Conflux}

\newcommand{\phvname}[1]{{\fontfamily{phv}\selectfont #1}}

\newtheorem{lemma}{Lemma}
\newtheorem{theorem}[lemma]{Theorem}

\newenvironment{proof}{\vspace{-0.05in}\noindent{\bf Proof.}}%
{\hspace*{\fill}$\Box$\par}
{\hspace*{\fill}$\Box$\par\vspace{4mm}}
{\hspace*{\fill}$\Box$\par}

\newcommand{\lh}{\lambda_{\text{h}}} 
\newcommand{\la}{\lambda_{\text{a}}} 

\newcommand\blfootnote[1]{%
  \begingroup
  \renewcommand\thefootnote{}\footnote{#1}%
  \addtocounter{footnote}{-1}%
  \endgroup
}

\begin{document}
\title{\bf Scaling Nakamoto Consensus to Thousands of Transactions per Second}
\author{Chenxing Li*, Peilun Li*, 
Dong Zhou*$^\dagger$, 
Wei Xu, Fan Long$^\dagger$, and Andrew Chi-Chih Yao\\
Institute for Interdisciplinary Information Sciences, Tsinghua University\\
Alt-chain Technologies$^\dagger$
}
\date{}
\maketitle
\thispagestyle{empty}

\begin{abstract}
This paper presents {\name}, a fast, scalable and decentralized blockchain
system that optimistically process concurrent blocks without discarding any as
forks. The {\name} consensus protocol represents relationships between blocks
as a direct acyclic graph and achieves consensus on a total order of the
blocks. {\name} then, from the block order, deterministically derives a
transaction total order as the blockchain ledger. We evaluated {\name} on
Amazon EC2 clusters with up to 20k full nodes. {\name} achieves a transaction
throughput of 5.76GB/h while confirming transactions in 4.5-7.4 minutes. The
throughput is equivalent to 6400 transactions per second for typical Bitcoin
transactions. Our results also indicate that when running {\name}, the
consensus protocol is no longer the throughput bottleneck. The bottleneck is
instead at the processing capability of individual nodes. 

\blfootnote{*Equal contributions and ranked alphabetically.}
\blfootnote{$^\dagger$Starting from September 2018, Fan Long is an assistant professor at University of
    Toronto and Dong Zhou is a PhD student at Carnegie Mellon University. }
\end{abstract}

\section{Introduction}
\label{sec:introduction}

Following the success of the cryptocurrencies~\cite{bitcoin,
  Ethereum}, blockchain has recently evolved into a technology
platform that powers secure, decentralized, and consistent transaction
ledgers at Internet-scale. The ledger becomes a powerful abstraction
and fuels innovations on real-world applications in financial systems,
supply chains, and health cares~\cite{IBMSupplyChain,
  DeloitteFinancial, DeloitteHealthCare}, shifting the landscapes of
the industries that worth hundreds of billions of dollars.

Blockchain platforms like Bitcoin~\cite{bitcoin} use \emph{Nakamoto consensus}
as their consensus protocols. This protocol usually organizes
transactions into an ordered list of blocks (i.e., a blockchain), each of which
contains multiple transactions and a link to its predecessor. To defend against
Sybil attacks, everyone solves \emph{proof-of-work} problems
(e.g., finding partial hash collisions) to compete for the right of generating
the next block. To prevent an attacker from reverting previous transactions,
everyone agrees on the longest chain of blocks as the correct transaction
history. Each newly generated block will be appended at the end of the longest
chain to make the chain even longer and therefore harder to revert.

However, the performance bottleneck remains one of the most critical challenges
of current blockchains. In the standard Nakamoto consensus, the performance
is bottlenecked by the facts 1) that only one participant can win the
competition and contribute to the blockchain, i.e., concurrent blocks
are discarded as \emph{forks}, and 2) that the slowness is essential to
defend against adversaries~\cite{GHOST, SPECTRE}. For example, Bitcoin
generates one 1MB block every 10 minutes and can therefore only process 7
transactions per second. Furthermore, to obtain high confidence that a
transaction is irreversible in Bitcoin, users typically have to wait for
tens of blocks building on top of the enclosing block of the transaction.
This causes hours of waiting before confirming a transaction. The insufficient
throughput and long confirmation delay severely limit the adoptions of
blockchain techniques, causing poor user experience, congested network, and
skyrocketing transaction fees~\cite{Ethereum-CryptoKitties, Bitcoin-Fee}.

Previous research focuses on reducing the participations of the consensus to
improve the performance without compromising the security of the blockchains.
For example, Bitcoin-NG~\cite{BitcoinNG} periodically elects a leader and
allows the leader to dictate the transaction total order for a period of time.
It improves the throughput but not the confirmation time of transactions.
Several proposals elect a small set of participants as a committee to run
Byzantine fault tolerance (BFT) to determine the transaction
order~\cite{ByzCoin, HoneyBadgerBFT, stellar, Algorand}. This solution may
create undesirable hierarchies among protocol participants and compromise the
decentralization of the blockchains.

\subsection{{\name}}

We present {\name}, a fast, scalable, and decentralized blockchain system that
can process thousands of transactions per second while confirming each
transaction in minutes. The core of {\name} is its consensus protocol that
allows multiple participants to contribute to the {\name} blockchain
concurrently while still being provably safe. The protocol enables
significantly faster block generation and, in turn, enables higher throughputs
and faster confirmations. In fact, the {\name} throughput is no longer limited
by the consensus protocol, but by the processing capability of each individual
node.


One observation that drives the design of {\name} is that standard Nakamoto
consensus protocols preemptively define a restrictive total order of
transactions when generating each block. This strategy introduces many false
dependencies that lead to unnecessary forks. The key to our approach is to
\emph{defer the transaction total ordering and optimistically process
    concurrent transactions and blocks}. Given that the transactions rarely
conflict in blockchains (particularly in cryptocurrencies), {\name} first
optimistically assumes that transactions in concurrent blocks would not
conflict with each other by default and therefore considers only explicit
happens-before relationships specified by the block generators. This enables
{\name} to operate with less constraints and efficiently achieve consensus on
one block total order among many possibilities. Each participant in {\name}
then deterministically derives the transaction total order from the agreed
block total order, discarding conflicting and/or duplicated transactions. This
lazy reconciliation provides the same external interface to the users.

To safely incorporate contributions from concurrent blocks, the {\name}
consensus protocol maintains two kinds of relationships between blocks.  When a
participant node generates a new block in {\name}, the node identifies a parent
(predecessor) block for the new block and creates a \emph{parent edge} between
these two blocks like Bitcoin. These parent edges enable {\name} to achieve
consistent irreversible consensus on its ledger.  The node also identifies all
blocks that have no incoming edge and creates \emph{reference edges} from the
new block to those blocks. Such reference edges represent that those blocks are
generated before the new block.  They enable {\name} to systematically extend
the achieved consensus to incorporate concurrent blocks.  As a result, the
edges between blocks now form a sophisticated direct acyclic graph (DAG) rather
than a chain with potential forks~\cite{SPECTRE, PHANTOM}.

One challenge {\name} faces is how to agree on an irreversible block total
order of the DAG that include concurrent blocks. {\name} addresses this
challenge with its novel ordering algorithm. Given a \emph{pivot chain} that
starts from the genesis block and that contains only parent edges, the ordering
algorithm deterministically partitions all blocks in the DAG into \emph{epochs}
using the pivot chain, topological sorts blocks in each epoch, and concatenates
the sorting results of all epochs into the final total order following the
happens-before order between epochs. The algorithm guarantees that if the pivot
chain stabilizes (i.e., blocks on the chain are irreversible except the last
few blocks), then the produced block total order will stabilize (i.e., blocks
in the total order are irreversible except blocks in the last few epochs). With
its ordering algorithm, {\name} reduces the problem of achieving consensus on a
block total order of the DAG to the problem of achieving consensus on a pivot
chain. {\name} therefore uses a modified chain-based Nakamoto
consensus~\cite{GHOST} to solve the chain consensus problem.

\noindent \textbf{Assumptions and Guarantees:}
{\name} operates under similar assumptions as Bitcoin and many others. 
Honest nodes are reasonably synchronous respecting a network
diameter $d$ (i.e., messages sent by honest nodes will be delivered with a
maximum delay of $d$). All honest nodes together control more block generation
power than attackers. Under these assumptions, the {\name} consensus
protocol guarantees that the agreed block total order is irreversible with very
high probability.

Note that we focus on the consensus protocol design and implementation and the
incentive mechanism is outside the scope of this paper. Therefore in this
paper, we refer ``honest nodes'' as nodes that run bug-free programs that
faithfully implement {\name}. We leave designing a compatible incentive
mechanism to economically encourage honest behaviors as future work.

\subsection{Experimental Results}

We implemented a prototype of {\name} and evaluated {\name} by
deploying up to 20k {\name} full nodes on 800 Amazon EC2 virtual machines.
We also compared {\name} with two standard chain-based Nakamoto consensus
approaches, Bitcoin~\cite{bitcoin} and GHOST~\cite{GHOST}. 

Our experimental results show that under the 20Mbps bandwidth limit for each
full node (i.e., the same experimental environment as
Algorand~\cite{Algorand}), {\name} can process one 4MB block per 5 seconds and
therefore achieves the transaction throughput of 2.88GB/h; the {\name}
throughput is 11.62x of the throughput of Bitcoin and GHOST, and is 3.84x of the
throughput of Algorand~\cite{Algorand}. Under the 40Mbps bandwidth limit,
{\name} can process one 4MB block per 2.5 seconds and therefore achieves an
even higher transaction throughput of 5.76GB/h. 
The {\name} throughput is equivalent to 6400 transactions per second for
typical Bitcoin transactions. 
In fact, our results indicate
that when running {\name}, the consensus protocol is no longer the bottleneck
of the throughput, but the processing capability of
individual nodes (e.g., the hardware bandwidth limit).

Our experiments also show that by working with faster generation rates, {\name}
allows blocks building on top of each other more quickly and therefore enables
fast confirmations. {\name} confirms transactions in 4.5-7.4 minutes under the
4MB/2.5s and 40Mbps setting and 7.6-13.8 minutes under the 4MB/5s and 20Mbps
setting.


\subsection{Contribution}

This paper makes the following contributions:

\begin{itemize}
\item \textbf{Consensus Protocol:} We present a fast and scalable DAG-based
    Nakamoto consensus protocol and its prototype implementation, {\name}, to
    optimistically process concurrent blocks while lazily reconciling the
    transaction total order from an agreed block total order. {\name} novelly
    maintains two different kinds of relationships between generated blocks to
    safely incorporate contributions from concurrent blocks into its ledger.




\item \textbf{{\name} Implementation:} We present a prototype implementation of
    {\name} based on the Bitcoin core codebase~\cite{Bitcoin-core}. To the best
    of our knowledge, {\name} is the first blockchain system that uses a
    DAG-based Nakamoto consensus protocol and that can process thousands of
    transactions per second.


\item \textbf{Experimental Results:} We present a systematic large-scale
    evaluation of {\name}. Our results show that, when running with up to 20k
    full nodes, {\name} can achieve the transaction throughputs of 2.88GB/h and
    5.76GB/h and confirm transactions with high confidence in 4.5-13.8 minutes.
\end{itemize}

The rest of this paper is organized as follows. Section~\ref{sec:overview}
presents an overview of {\name} and an illustrative running example for
{\name}. Section~\ref{sec:algorithm} presents the {\name} consensus protocol.
Section~\ref{sec:implementation} discusses our prototype implementation of
{\name}. Section~\ref{sec:results} presents our evaluation of {\name}. We
discuss related work in Section~\ref{sec:related_work} and conclude in
Section~\ref{sec:conclusion}.


\section{Overview and Example}
\label{sec:overview}

This section presents an overview of {\name} and a running example to
illustrate the {\name} consensus protocol.

\subsection{{\name} Architecture}

\begin{figure}
\centering
\includegraphics[width=0.95\columnwidth]{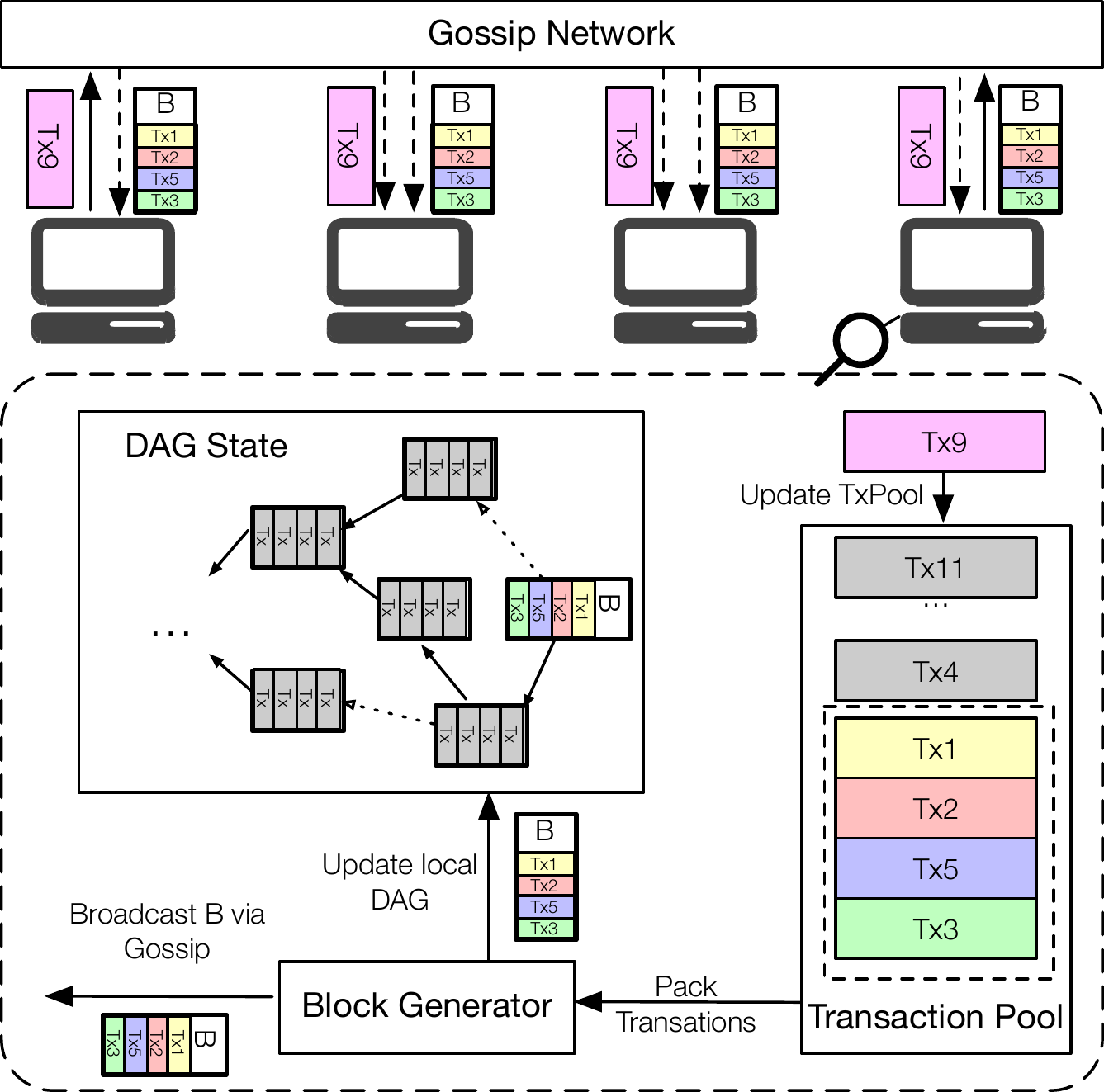}
\caption{Architecture of {\name}.}
\label{fig:architecture}
\end{figure}

Figure~\ref{fig:architecture} presents the architecture diagram of {\name}.
Similar to Bitcoin, a transaction in {\name} is a payment message
signed by the payer to transfer coins from the payer to the
payee, where the payer and the payee are identified by their public
keys. There can be special coinbase transactions to mint new coins.
Each transaction also has a unique id generated by cryptographic
digest functions to ensure its integrity. A block consists of a list of
transactions (e.g. there are 4 transactions in block \phvname{B} of
Figure~\ref{fig:architecture}) and reference links to
previous blocks (e.g., in the DAG state in Figure~\ref{fig:architecture},
\phvname{B} links to two previous blocks). Each block also has a unique id
generated by digest functions to ensure data integrity. At the very beginning,
{\name} starts with a predefined \emph{genesis block} to determine the initial
state of the blockchain.

A key difference between {\name} and traditional blockchains is that
the blocks and the edges form a DAG instead of a linear chain. As the
blocks travel across the netowrk, each node might observe a slighlty
different DAG due to network delays. The goal of {\name} is to
maintain the local DAG of each individual node so that all nodes 
in {\name} can eventually agree on a total order of blocks (and transactions).

Note that having total orders in {\name} enables the supports of smart
contracts, which execute Turing complete 
programs to update the state associated with accounts. For the purpose of
illustration, in this paper we assume a 
balance model with simple payment transactions. One transaction
conflicts with previous transactions if after previous transactions,
the payer does not have enough balance to execute the transaction.

\noindent \textbf{Gossip Network:} All participant nodes in {\name} are
connected via a gossip network as shown in the top part of
Figure~\ref{fig:architecture}. Whenever a node initiates a transaction, it
broadcasts the transaction to all other nodes via the gossip network (e.g.,
\phvname{Tx9} in Figure~\ref{fig:architecture}). Whenever a node generates a
new block, it broadcasts the block to all other nodes via the network
as well (e.g., the block \phvname{B} in Figure~\ref{fig:architecture}). {\name}
currently uses a modified gossip network implementation in Bitcoin
Core~\cite{Bitcoin-core} (see Section~\ref{sec:implementation}).

\noindent \textbf{Pending Transaction Pool:} Each node maintains a
pending transaction pool as shown in the bottom right of
Figure~\ref{fig:architecture}. The pool contains all transactions
that have been heard by the node but are not yet packed into any block.
Whenever a node receives a new transaction from the gossip network, the node
adds the transaction into its pool (e.g., adding \phvname{Tx9} into
the transaction pool in Figure~\ref{fig:architecture}). Whenever a node
discovers a new block (either by generating the block or by receiving the block
from other nodes), the node removes all transactions in the new block from its
pending transaction pool. For example, in Figure~\ref{fig:architecture}, the
node removes \phvname{Tx1}, \phvname{Tx2}, \phvname{Tx3}, and \phvname{Tx5}
from its pool after generating the block \phvname{B}. In {\name}
concurrent blocks might pack duplicate or conflicting transactions due
to network delays and malicious nodes. These transactions will be resolved by our consensus protocol.

\noindent \textbf{Block Generator:} Each node in {\name} runs a block generator
to generate valid new blocks to pack pending transactions as shown in the
bottom of Figure~\ref{fig:architecture}. {\name} operates with proof-of-work
(PoW) mechanism, so the block generator attempts to find solutions for PoW
problems to generate blocks. A valid block header must contain a PoW solution
in its header. Once such a valid block header is generated, it selects pending
transactions from the pool to fill the new block. Similar to Bitcoin, the PoW
mechanism maintains a stable network block generation rate via dynamically
adjusting the difficulty of the PoW problems. Note that besides PoW, {\name}
can also work with any other mechanism that can maintain a stable block
generation rate, such as proof-of-stake (PoS)~\cite{Casper, PeerCoin}.



\noindent \textbf{Local DAG State: } Each node maintains a local state that
contains all blocks which the node is aware of. Because in {\name} each block
may contain links to reference several previous blocks not just one, the result
state is a direct acyclic graph (DAG) as shown in center of
Figure~\ref{fig:architecture}. Whenever the node discovers a new block (either
by generating or by receiving it), the
node updates its local DAG state accordingly.

\subsection{Consensus Protocol}

The {\name} consensus protocol operates on the local DAG state of each
individual node. The goal of the protocol is to achieve the
consensus on a total order of generated blocks among all
nodes. From the total order of the blocks {\name} then derives a total order of
transactions inside those blocks to process these transactions. For conflicting
or duplicated transactions, {\name} only processes the first one and discards
the remaining ones.



\begin{figure*}
\centering
\includegraphics[scale=0.52]{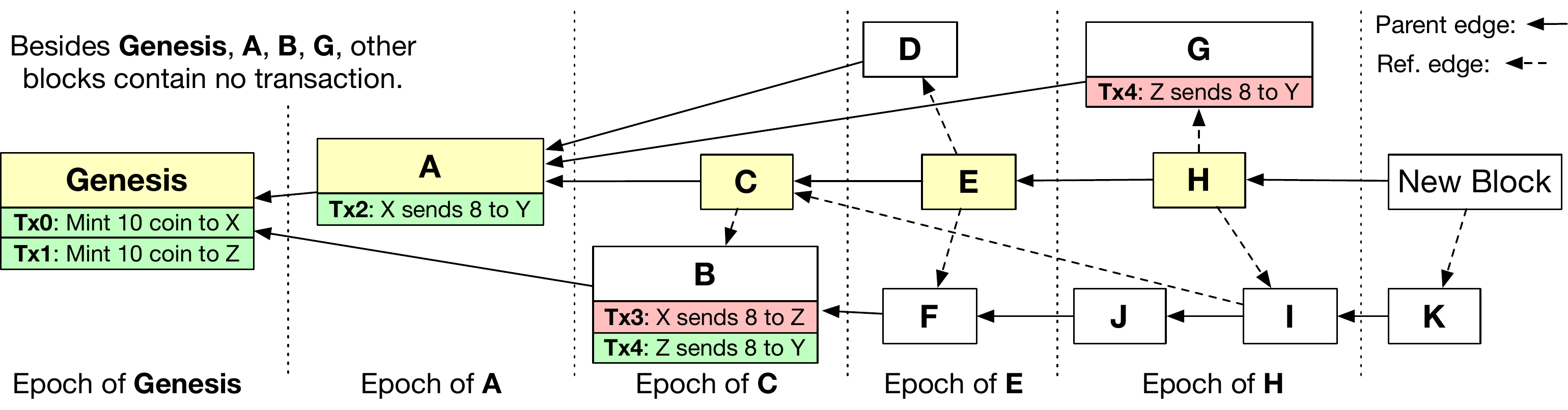}
\caption{An example local DAG state to illustrate the consensus algorithm of
{\name}. The yellow blocks are on the pivot chain in the DAG. 
Each block on the pivot chain forms a new epoch to partition
blocks in the DAG.}
\label{fig:example}
\end{figure*}

\noindent \textbf{DAG and Edges:} Figure~\ref{fig:example} presents a running
example of the local DAG state of a node in {\name}. We will use this example
in the remaining of this section to illustrate the high level ideas of the
{\name} consensus protocol. Each vertex in the DAG in
Figure~\ref{fig:example} corresponds to a block. In Figure~\ref{fig:example},
\phvname{Genesis} is the predefined genesis block. Only \phvname{Genesis},
\phvname{A}, \phvname{B}, and \phvname{G} are associated with transactions.
There are two kinds of edges in the DAG, parent edges and reference edges:

\begin{itemize}
\item \textbf{Parent Edge:} Each block except \phvname{Genesis} has exactly
one outgoing parent edge (solid line arrows in Figure~\ref{fig:example}). Intuitively, the parent edge corresponds to a voting relationship,
i.e., the node that generates the child block votes for the transaction history
represented by the parent block.
For example, there are one parent edge from \phvname{C} to \phvname{A} and one
from \phvname{F} to \phvname{B}.

\item \textbf{Reference Edge:} Each block can have multiple outgoing reference
edges (dashed lines arrows in Figure~\ref{fig:example}).
A reference edge corresponds to generated-before relationships
between blocks. For example, there is a reference edge from \phvname{E} to
\phvname{D}. It indicates that \phvname{D} is generated before \phvname{E}.
\end{itemize}

\noindent \textbf{Pivot Chain:}
Note that all parent edges in a DAG together form a \emph{parental tree} in
which the genesis block is the root. In the parental tree, {\name} selects a
chain from the genesis block to one of the leaf blocks as the \emph{pivot
chain}. 
Unlike the Bitcoin protocol which selects the longest chain in the
tree, {\name} selects the pivot chain based on the GHOST rule~\cite{GHOST}. Specifically, the
selection algorithm starts from the genesis block. At each step, it computes
the subtree sizes of each child in the parental tree and advances to the child
block with the largest subtree, until it reaches a leaf block. The advantage of
the GHOST rule is that it guarantees the
irreversibility of the selected pivot chain even when facing forks of honest
nodes due to network delays as the blocks in the forks also contribute
to the safety of the pivot chain according to the analysis~\cite{GHOST}.

In Figure~\ref{fig:example}, {\name} selects \phvname{Genesis}, \phvname{A},
\phvname{C}, \phvname{E}, and \phvname{H} as the pivot chain. Note that this is
not the longest chain in the parental tree and the longest chain is
\phvname{Genesis}, \phvname{B}, \phvname{F}, \phvname{J}, \phvname{I}, and
\phvname{K}. {\name} does not select this longest chain because the subtree of
\phvname{A} contains more blocks than the subtree of \phvname{B}. Therefore,
the chain selection algorithm selects \phvname{A} over \phvname{B} at its first
step.


\noindent \textbf{Generating New Block:} Whenever a node generates a new block,
it first computes the pivot chain in its local DAG state and sets the last
block in the chain as the parent of the new block. This makes the chain heavier
(i.e., the corresponding subtree of each block on the chain contains one extra
block), so that other nodes will be even more likely to select the same chain
as their pivot chains in future. The node then finds all tip blocks in the DAG
that have no incoming edge and creates reference edges from the new block to
each of those tip blocks. For example, if the node has a local state as shown
in Figure~\ref{fig:example}, when generating a new block, the node will choose
\phvname{H} as the parent of the new block and will create a reference edge
from the new block to \phvname{K}.

\noindent \textbf{Epoch:} 
Parent edges, reference edges, and the pivot chain together enable {\name} to
split all blocks in a DAG into epochs. As shown in Figure~\ref{fig:example},
every block in the pivot chain corresponds to one epoch. Each epoch contains
all blocks 1) that are reachable from the corresponding block in the pivot
chain via the combination of parent edges and reference edges and 2) that are
not included in previous epochs. For example, in Figure~\ref{fig:example},
\phvname{J} belongs to the epoch of \phvname{H} because \phvname{J} is
reachable from \phvname{H} but not reachable from the previous pivot chain
block, \phvname{E}.

\noindent \textbf{Block Total Order:} {\name} determines the total order of
the blocks in a DAG as follows. {\name} first sorts the blocks based on their
corresponding epochs and then sorts the blocks in each epoch based on their
topological order. If two blocks in an epoch have no partial order
relationship, {\name} breaks ties deterministically with the unique ids of the
two blocks. For the local DAG in Figure~\ref{fig:example}, {\name} obtains the
total order as the following: \phvname{Genesis}, \phvname{A}, \phvname{B},
\phvname{C}, \phvname{D}, \phvname{F}, \phvname{E}, \phvname{G}, \phvname{J},
\phvname{I}, \phvname{H}, and \phvname{K}.

\noindent \textbf{Transaction Total Order:} {\name} first sorts
transactions based on the total orders of their enclosing blocks. If two transactions belong to the
same block, {\name} sorts the two transactions based on the appearance order in
the block.

{\name} checks the conflicts of the transactions at the same time when deriving
the orders. If two transactions are conflicting with each other, {\name} will
discard the second one. If one transaction appears in multiple blocks, {\name}
will only keep the first appearance and discard all redundant ones. In
Figure~\ref{fig:example}, the transaction total order is \phvname{Tx0},
\phvname{Tx1}, \phvname{Tx2}, \phvname{\st{Tx3}}, \phvname{Tx4}, and
\phvname{\st{Tx4}}. {\name} discards \phvname{Tx3} because it conflicts with
\phvname{Tx2}.

\subsection{Security Analysis}

We next discuss potential attack strategies and explain why {\name} is safe
against these attacks. For a more formal discussion of the {\name} safety, see
Section~\ref{sec:algorithm:correctness}. 

Now suppose an attacker wants to revert the transaction \phvname{Tx4} in the
block \phvname{B} in Figure~\ref{fig:example}. To do so, the attacker needs to
revert the agreed total order of the blocks so that the attacker can insert an
attacker block before \phvname{B} and that the attacker block contains a
transaction which conflicts with \phvname{Tx4}. One naive attack strategy is to
link such an attacker block to existing blocks in very early epochs, e.g.,
setting \phvname{Genesis} as the parent. However, because the attacker block is
new with no children, it will not become a pivot chain block. By our epoch
definition, the attacker block has to wait for the references of future pivot
chain blocks. The block will still belong to a future epoch despite the fact
that it sets \phvname{Genesis} as its parent. Therefore the block will not
appear before \phvname{B} in the total order.

Therefore to revert a transaction enclosed by a block, the attacker has to
revert the partition scheme of epochs before the block. Because the epoch
partition scheme is deterministically defined based on the pivot chain, the
attacker therefore has to revert the corresponding blocks on the pivot chain.
If the attacker attempts reverting the pivot chain, the situation is then similar
to the double spending attack in chain-based Nakamoto consensus protocol. In
this situation, the safety property of {\name} relies on the fact that all
honest nodes will continue to work on the pivot chain to make the pivot chain
longer and heavier. Since honest nodes together have more block generation
power than the attacker, as the time passes by, early blocks on the pivot chain
are increasingly irreversible for the attacker.

\noindent \textbf{Transaction Confirmation: } A user can confirm his or her
transaction with the following strategy. The user locates the first epoch that
contains a block including the transaction. The user identifies the
corresponding pivot chain block of the epoch. The user then decides how much
risk he or she can tolerate based on the estimations of the block generation
power that the attacker controls. The user finally estimates the risk of the pivot chain block being
reverted using the formula in Section~\ref{sec:algorithm:correctness} to decide
whether to confirm the transaction.


\section{{\name} Consensus}
\label{sec:algorithm}

In this section we present the {\name} consensus algorithm and discuss its safety and
liveness properties.

\subsection{Consensus Algorithm}

At any time, the local state of a user in the {\name} protocol is a graph $G =
\langle B, g, P, E\rangle$. $B$ is the set of blocks in $G$. $g \in B$ is the
genesis block. $P$ is a function that maps a block $b$ to its parent block
$P(b)$. Specially, $P(g) = \bot$. $E$ is the set of directed reference edges
and parent edges in this graph. $e = \langle b, b' \rangle \in E$ is an edge
from the block $b$ to the block $b'$, which denotes that $b'$ happens before $b$.
Note that there is always a parent edge from a block to its parent block (i.e.,
$\forall b \in B, \langle b, P(b) \rangle \in E$).
All nodes in the {\name} protocol share a predefined deterministic hash
function $\mathrm{Hash}$ that maps each block in $B$ to a unique integer id. It
satisfies that $\forall b \neq b', \mathrm{Hash}(b) \neq \mathrm{Hash}(b')$. 

\begin{figure}
\small
\framebox{
\begin{minipage}{1.5in}
\begin{displaymath}
G = \langle B, g, P, E \rangle
\end{displaymath}
\end{minipage}
}
\begin{displaymath}
\begin{array}{l}
\mathrm{Chain}(G, b) = \left\{
\begin{array}{ll}
g & b = g\\
\mathrm{Chain}(G, P(b)) \circ b & \text{otherwise} \\
\end{array}
\right.\\
\mathrm{Child}(G, b) = \{ b' \mid P(b') = b \} \\
\mathrm{Sibling}(G, b) = \mathrm{Child}(G, P(b)) - \{b\}\\ 
\mathrm{Subtree}(G, b) = (\cup_{i \in \mathrm{Child}(G, b)}{\mathrm{Subtree}(G, i)}) \cup \{b\} \\
\mathrm{Before}(G, b) = \{b' \mid b' \in B, \langle b, b' \rangle \in E\} \\
\mathrm{Past}(G, b) = (\cup_{i \in \mathrm{Before}(G, b)} {\mathrm{Past}(G, i)}) \cup \{b\} \\
\mathrm{TotalOrder}(G) = \mathrm{ConfluxOrder}(G, \mathrm{Pivot}(G, g))
\end{array}
\end{displaymath}
\caption{The Definitions of $\mathrm{Chain}()$, $\mathrm{Child}()$, $\mathrm{Sibling}()$, $\mathrm{Subtree}()$,
$\mathrm{Before}()$, $\mathrm{Past}()$, 
and $\mathrm{TotalOrder}()$.}
\label{fig:utils}
\end{figure}

We next define several utility functions and notations. $\mathrm{Chain}()$
returns the chain from the genesis block to a given block following only parent
edges. $\mathrm{Child}()$ returns the set of child blocks of a given block.
$\mathrm{Sibling}()$ returns the set of siblings of a given block.
$\mathrm{Subtree}()$ returns the sub-tree of a given block in the parental
tree. $\mathrm{Before}()$ returns the set of blocks that are immediately
generated before a given block. $\mathrm{Past}()$ returns the set of blocks
that are generated before a given block (but including the block itself). 
Figure~\ref{fig:utils} presents the definition of these utility functions. In
the rest of this section, we use ordered lists to denote chains and serialized
orders. ``$\circ$'' denotes the concatenation of two ordered lists.

\begin{algorithm}[t]
\small
\SetNlSty{}{}{}
\DontPrintSemicolon
\SetKw{To}{ to }
\SetKw{In}{ in }
\SetKw{Or}{ or }
\SetKw{And}{ and }
\SetKwInOut{Input}{Input}\SetKwInOut{Output}{Output}
\SetKw{Define}{define}
\Input{The local state $G = \langle B, g, P, E \rangle$ and a starting block $b \in B$}
\Output{The last block in the pivot chain for the subtree of $b$ in $G$}
\let\oldnl\nl
\newcommand{\nonl}{\renewcommand{\nl}{\let\nl\oldnl}}

\If {$\mathrm{Child}(G, b) = \emptyset$} {
    \Return $b$\;
}
\Else {
    $s \longleftarrow \bot$ \;
    $w \longleftarrow -1$\;
    \For {$b' \in \mathrm{Child}(G, b)$}{
        $w' \longleftarrow |\mathrm{Subtree}(G, b')|$\;
        \If {$w' > w$ \Or $w' = w$ \And $\mathrm{Hash}(b') < \mathrm{Hash}(s)$} {
            $w \longleftarrow w'$\;
            $s \longleftarrow b'$\;
        }
    }
    \Return $\mathrm{Pivot}(G, s)$\;
}
\caption{The definition of $\mathrm{Pivot}(G, b)$.}
\label{fig:pivot}
\end{algorithm}

\noindent \textbf{Pivot Chain Selection: } Figure~\ref{fig:pivot}
presents our pivot chain selection algorithm (i.e., the definition of
$\mathrm{Pivot}()$). Given a DAG state $G$,
$\mathrm{Pivot}(G, g)$ returns the last block in the pivot chain starting
from the genesis block $g$. The algorithm recursively advances to the child
block whose corresponding subtree has the largest number of blocks (lines
4-10). When there are multiple child blocks with the same number, the algorithm
selects the child block with the smallest unique hash id (line 8). The
algorithm terminates until it reaches a leaf block (lines 1-2).




\begin{algorithm}[t]
\small
\SetNlSty{}{}{}
\DontPrintSemicolon
\SetKw{To}{ to }
\SetKw{In}{ in }
\SetKw{And}{ and }
\SetKwInOut{Input}{Input}
\SetKwInOut{State}{Local State}
\SetKwFor{Event}{upon event}{do}{}
\State{A graph $G = \langle B, g, P, E \rangle$}
\let\oldnl\nl
\newcommand{\nonl}{\renewcommand{\nl}{\let\nl\oldnl}}
\While{\text{Node is running}}{
    \Event{\text{Received $G'= \langle B', g, P', E' \rangle$}}{
        $G'' \longleftarrow \langle B \cup B', g, P \cup P', E \cup E' \rangle$\;
        \If {$G \neq G''$}{
            $G \longleftarrow G''$\;
            \text{Broadcast the updated $G$ to other nodes}\;
        }
    }
    \Event{\text{Generated a new block $b$}}{
        $a \longleftarrow \mathrm{Pivot}(G, g)$\;
        $E' \longleftarrow E \cup \{\langle b, t \rangle \mid \forall b' \in B, \, \langle b', t \rangle \notin E\}$\;
        $G \longleftarrow \langle B \cup \{b\}, g, P[b \mapsto a], E' \rangle$\;
        \text{Broadcast the updated $G$ to other nodes}\;
    }
}
\caption{The {\name} Main Loop}
\label{fig:main-loop}
\end{algorithm}

\noindent \textbf{Consensus Main Loop:} Figure~\ref{fig:main-loop} presents the
main loop of a node running the consensus algorithm. It processes two kinds of
events. The first kind of events is receiving DAG update information from other
nodes via the underlying gossip network. The node updates its local state
accordingly (lines 2-3) and relays the local state via the gossip network (lines
4-6). Note that here we assume that the underlying gossip network already
verified the integrity of the received information (i.e., cryptographic
signatures in the blocks).

The second kind of events corresponds to the local block generator successfully
generating a new block $b$. In this case, the node adds $b$ into its local DAG
$G$ and updates $G$ as follows. It first sets the last block of the pivot chain
in its local DAG as the parent of $b$ (lines 8 and 10). The node also finds all
of those blocks in the local DAG that have no incoming edges and creates a
reference edge from $b$ to each of those blocks (line 9). The node finally
broadcast the updated $G$ to other nodes via the gossip network. Note that $P[b
\mapsto a]$ at line 10 denotes the result map of mapping $b$ to $a$ in $P$ (other mappings
in $P$ are unchanged).

Note that for brevity, the pseudo-code in Figure~\ref{fig:main-loop} broadcasts
the whole graph to the network. In our {\name} implementation, {\name}
broadcasts and relays each individual block to avoid unnecessary network
transmissions. See Section~\ref{sec:implementation} for the details.

\begin{algorithm}[t]
\small
\SetNlSty{}{}{}
\DontPrintSemicolon
\SetKw{To}{ to }
\SetKw{In}{ in }
\SetKw{And}{ and }
\SetKwInOut{Input}{Input}\SetKwInOut{Output}{Output}
\Input{The local state $G = \langle B, g, P, E \rangle$ and a block $a$}
\Output{A list of blocks $L = b_1 \circ b_2 \circ \ldots \circ b_n$, where $b_1 = g$ and $\forall 1 \le i \le n, b_i \in B$}
\let\oldnl\nl
\newcommand{\nonl}{\renewcommand{\nl}{\let\nl\oldnl}}

$a' \longleftarrow P(a)$\;
\If {$a' = \bot$} {
    \Return $a$\;
}
$L \longleftarrow \mathrm{ConfluxOrder}(G, a')$\;
$B_{\Delta} \longleftarrow \mathrm{Past}(G, a) - \mathrm{Past}(G, a')$\;
\While {$B_{\Delta} \neq \emptyset$} {
$G' \longleftarrow \langle B_{\Delta}, g, P, E \rangle$\;
$B_{\Delta}' \longleftarrow \{x \mid |\mathrm{Before}(G', x)| = 0\}$\;
Sort all blocks in $B_{\Delta}'$ in order as $b'_1, b'_2, \ldots, b'_k$ \\
\quad such that $\forall 1 \leq i < j \leq k, \, \mathrm{Hash}(b'_i) < \mathrm{Hash}(b'_j)$\;
$L \longleftarrow L \circ b'_1 \circ b'_2 \circ \ldots \circ b'_k$\;
$B_{\Delta} \longleftarrow B_{\Delta} - B_{\Delta}'$\;
}
\Return $L$\;
\caption{The Definition of $\mathrm{ConfluxOrder}()$.}
\label{fig:order}
\end{algorithm}

\noindent \textbf{Total Order:} Figure~\ref{fig:order} defines
$\mathrm{ConfluxOrder}()$, which corresponds to our block ordering algorithm.
Given the local state $G$ and a block $a$ in the pivot chain,
$\mathrm{ConfluxOrder}(G, a)$ returns the ordered list of all blocks that appear
in or before the epoch of $a$. Using $\mathrm{ConfluxOrder}()$, the total order
of a local state $G$ is defined as $\mathrm{TotalOrder}(G)$ in Figure~\ref{fig:utils}.

The algorithm in Figure~\ref{fig:order} first recursively orders all
blocks in previous epochs (i.e., the epoch of $P(a)$ and before). It then 
computes all blocks in the epoch of $a$ as $B_{\Delta}$ (line 5). It
topologically sorts all blocks in $B_{\Delta}$ and appends it into the result
list (lines 6-12). The algorithm uses the unique hash id to break
ties (lines 8-9).

\subsection{Assumptions and Parameters}
\label{assumption}

We next present the protocol assumptions that are important to our discussion
of the safety and liveness properties of the consensus algorithm.
Our protocol has the following assumptions and relevant parameters.

\noindent \textbf{Block Generation Rate: } The network together has a block
generation rate of $\lambda$. We use $\lh$ to denote the combined block
generation rate of honest nodes. We use $\lambda_{\text{a}} = q \cdot \lh$ to
denote the block generation rate of the attacker. Therefore $\lambda = \lh +
\lambda_{\text{a}}$ and $0 \le q < 1$.

\noindent \textbf{$d$-Synchronous: } We assume the underlying gossip network
provides $d$-synchronous communications for all honest nodes. If at time $t$,
one honest node broadcast a block or transaction via the gossip network, then
before time $t+d$, all honest nodes will receive this block and add this block
into their local states. 

\noindent \textbf{Adversary Model: } 
The attacker can choose arbitrary strategies to disrupt honest nodes. 
We assume two limitations to the capability of the attacker. Firstly, 
the attacker does not have the capability to reverse cryptographic functions.
Therefore honest nodes can reliably verify the integrity of a block in the
presence of the attacker. Secondly, all honest nodes combined
together have stronger block generation power than the attacker, i.e., as
we noted before $\lambda_{\text{a}} = q \cdot \lh$ and $0 \le q < 1$.


When the attacker generates a new block, the attacker can create parent and
reference edges from the new block to arbitrary existing blocks, not
necessarily following our protocol. Note that because all blocks in {\name} is
protected by cryptographic functions, the attacker however cannot modify edges
associated with an already generated block even if the block is generated by
the attacker himself/herself. The attacker can also withhold this new block for
a certain period of time before broadcasting the block. Because {\name}
implements a stale block detection mechanism similar to Bitcoin (see
Section~\ref{sec:implementation}), the attacker cannot withhold blocks forever
(but it can do so for hours). If the attacker decides to send a block to any
honest node, due to our $d$-synchronous assumption, all honest nodes will see
the block after a while.


\subsection{Correctness}
\label{sec:algorithm:correctness}

\noindent \textbf{Safety:} {\name} applies the GHOST rule to select its pivot
chain. The pivot chain of {\name} therefore satisfies the same safety property
as the selected chain in GHOST rule~\cite{GHOST}. As long as the attacker
controls less than half of the block generation power, it is highly unlikely to
revert an old common pivot chain block shared by all honest nodes~\cite{GHOST}.
Because all honest nodes will contribute toward the subtree of the pivot chain
block, after waiting a sufficiently long period of time it will be impossible
for the attacker to forge a subtree without the pivot chain block heavier than
the subtree of the pivot chain block~\cite{GHOST}. 

Because the algorithm in Figure~\ref{fig:pivot} is deterministic, once a prefix
of the pivot chain becomes stablized in all nodes (i.e., attacker cannot
revert), the algorithm will produce the same prefix of block orders for all
nodes.  This total order prefix is therefore irreversible because of the
irreversiblility of the pivot chain.

\noindent \textbf{Liveness: } 
As discussed above, the consensus of the block total order prefix in {\name}
depends on the consensus of the pivot chain prefix.  {\name} therefore has the
same liveness property as the GHOST protocol~\cite{GHOST}, i.e., eventually new
blocks will be appended to the common prefix of the pivot chain. 

\noindent \textbf{Confirmation: } For any block $b'$ in a DAG, suppose $b'$
belongs to the epoch of $b$, where $b$ is a pivot chain block. We can confirm
the transactions in $b'$ as long as $b$ becomes irreversible on the pivot
chain. Like standard Nakamoto consensus, although we can wait sufficiently long
to make the risk of an attacker reverting $b$ arbitrarily low. The risk will
always be greater than zero. We have the following lemma and theorem to bound
the risk of confirming $b$ on the pivot chain:

Suppose $b$ is a block on the pivot chain of all honest nodes during the
time $[t-d, t]$ and $P(b)$, the parent of $b$, is generated at time zero. The
chance of $b$ being kicked out of the pivot chain by one of its sibling blocks
$a$ is no more than:
\begin{displaymath}
\sum_{k=0}^{n-m}\zeta_k \cdot q^{n-m-k+1} + \sum_{k=n-m+1}^{\infty} \zeta_k
\end{displaymath}
where $n$ is the number of blocks in the subtree of $b$ before time $t-d$, $m$ is the number of blocks in the subtree of $a$ 
generated by honest nodes, and $\zeta_k = e^{-q\lh t}\frac{(q\lh t)^k}{k!}$.
This is a direct application of Theorem 10 in \cite{GHOST}.
It provides us a way to estimate the stability of each individual block on the
pivot chain. 

Note that the stability of a pivot
chain prefix is determined by the least stable block in the prefix.
Suppose $b$ is a block on the pivot chains of all honest nodes during the time $[t-d, t]$.
The chance of $b$ falls off the pivot chain is no more than:
\begin{displaymath}
\max_{\substack{a\in \mathrm{Chain}(G,b)\\ a'\in \mathrm{Sibling}(a)}} \mathrm{Pr}[\text{$a$ is kicked out of the pivot chain by $a'$}]
\end{displaymath}

\section{Implementation}
\label{sec:implementation}

We have implemented both {\name} and GHOST (for comparison in our evaluation)
based on the Bitcoin Core codebase v0.16.0~\cite{Bitcoin-core}.

\noindent \textbf{Block Header: }
To implement reference edges in {\name}, we modified the block header structure
in Bitcoin to include 32-byte block header hashes for each of its outgoing
reference edges.
Our experimental results show that this introduces a negligible overhead of
less than 960 bytes per block. Note that if a proof-of-work (PoW) scheme is
used to generate new blocks, these reference hashes must be included as part of
the puzzle to avoid having attackers be able to generate blocks with different
references at essentially zero-cost.

\noindent \textbf{Gossip Network: } 
Both {\name} and GHOST require to maintain the full structure of the tree/DAG
of blocks, while Bitcoin Core only propagates blocks in the identified longest
chain. We therefore modified the gossip network layer of Bitcoin Core to
broadcast and relay all blocks. To ensure that when a block is delivered to the
consensus layer, all of its past blocks are already delivered, {\name}
maintains the validity for each block. Whenever {\name} receives a block, it
traverses the DAG structure using breadth-first search (BFS) and updates the
validity of each traversed block. A block is valid if and only if {\name} has
received all past blocks (i.e., blocks that are reachable via parent and
reference edges) of the block. {\name} then delivers all of the newly validated
blocks to the consensus layer.

\noindent \textbf{Detecting Stale Blocks: }
Bitcoin Core has the following mechanism to detect stale blocks (e.g., blocks
generated and withheld by attackers). Each node periodically synchronizes with
its peers to maintain a network-adjusted time, which is the median of the
timestamps returned by its peers. Each block in Bitcoin Core is also
timestamped. A new block will be flagged as invalid if the timestamp of the new
block is earlier than the median timestamp of previous 11 blocks or if the
timestamp of the new block is two hours later than the network-adjusted time.

We use the same mechanism in Bitcoin Core with the following modifications.
First, we modify the rule proportionally to the block generation rates which we
use in our experiments. For example, if the system generates a block every 20
seconds (i.e., 30 times faster than Bitcoin), then a block is considered
invalid if its timestamp is earlier than the median timestamp of previous 330
blocks. Secondly, {\name} does not delete blocks with invalid timestamps. It
simply ignores this invalid block when counting the number of blocks in a
subtree for selecting the pivot chain. The rationale is that as long as the
invalid block no longer affects the partition scheme of already stable epochs,
it is safe to include the block into future epochs, processing transactions
inside the block.

\noindent \textbf{Bootstrapping a Node: }
When a node starts, it will handshake with each of its peers and run a
bootstrapping one-way synchronization process to update its local state. For
both GHOST and {\name}, the node needs to download all blocks in the tree/DAG
from its peers, not just the selected chain. To implement this one-way
synchronization, we enhanced the Bitcoin Core codebase with four extra message
types: \texttt{gettips}, \texttt{tips}, \texttt{getchains} and \texttt{chains}.
For sake of simplicity, let us consider the case where a node \textit{A}
attempts to download the blocks from another node \textit{B}. \textit{A} first
sends \textit{B} a \texttt{gettips} message requesting the list of tips, i.e.,
leaf blocks in the (parental) tree (represented by their 32-byte hashes).
\texttt{tips} is used in response to \texttt{gettips} to retrieve the block
hash of all the tips \textit{B} is aware of. For each of \textit{B}'s tip,
\textit{A} checks whether this is an unknown block or not and packs all the
answers in a \texttt{getchains} message. Upon receiving this \texttt{getchains}
message, for each new tip to \textit{A}, \textit{B} computes the last known
block in the chain from the genesis block to this tip and sends \textit{A} a
\texttt{chain} message containing the list of blocks starting right after the
last known block.

\section{Experimental Results}
\label{sec:results}

We next present a quantitative evaluation of {\name} to answer the following questions:
\begin{enumerate}
\item What is the throughput that {\name} can achieve for obtaining the block
total order?
\item What is the confirmation time that a user must wait in {\name} for
obtaining high confidence of irreversibility? How does this confirmation time
correlate with the block generation power that the attacker controls and
the risk that the user is willing to tolerate?
\item How does {\name} compare to previous chain-based Nakamoto consensus
protocols like Bitcoin~\cite{bitcoin} and GHOST~\cite{GHOST}?
\item How does {\name} scale as the network bandwidth changes? How does {\name}
scale as the number of full nodes grow?
\end{enumerate}

We deployed {\name} on up to 800 Amazon EC2 m4.2xlarge virtual machines (VM),
each of which has 8 cores and 1Gbps network throughput. By default, we run 25
{\name} full nodes in each VM and limit the bandwidth of each full node to
20Mbps. To model the network latency, we use the inter-city latency
measurements~\cite{citydelay} and assign each VM to one of 20 major cities. We
simulate the inter city delay by inserting artificial delays before the message
delivery. For each full node, the gossip network of {\name} connects it to an
average of 10 randomly selected peers. In our experiments, we assign all full
nodes with an equal block generation power. For each generated block, we use
the testing utilities in Bitcoin core code base to fill the block full with
artificial transactions. To avoid unnecessary PoW computation, we simulate the
mining with a poisson process.

We then deployed Bitcoin and our implementation of GHOST under the same
setup as {\name}. We run up to 20k full nodes in our
experiments. At April 2018, Bitcoin has fewer than 12k full
nodes~\cite{bitcoin-fullnode} and Ethereum~\cite{eth-fullnode} has fewer than
17k full nodes. Our experiments are at the same scale as those real-world
cryptocurrencies.

Note that in our experiments we measure the network diameter $d$ as the
propagation time for 99\% of the blocks to reach all full nodes. This enables
us to see the trend of network diameter (see
Section~\ref{sec:results:scalability}). There is always a few node generating
blocks when they are lagging behind, causing longer delays. Note that this
phenomena does not affect the correctness of {\name} because we can simply
count those nodes that lags behind as malicious nodes. Also note that For all
the experiments, we monitored the overhead introduced by {\name} due to its
DAG-based approach. The computation overhead is negligible compared to the PoW
puzzles, and the spatial overhead is a maximum of 960 bytes per block, which is
small compared to typical block size in several MBs.

\subsection{Throughput}

\begin{figure}[t]
	\includegraphics[width=\linewidth]{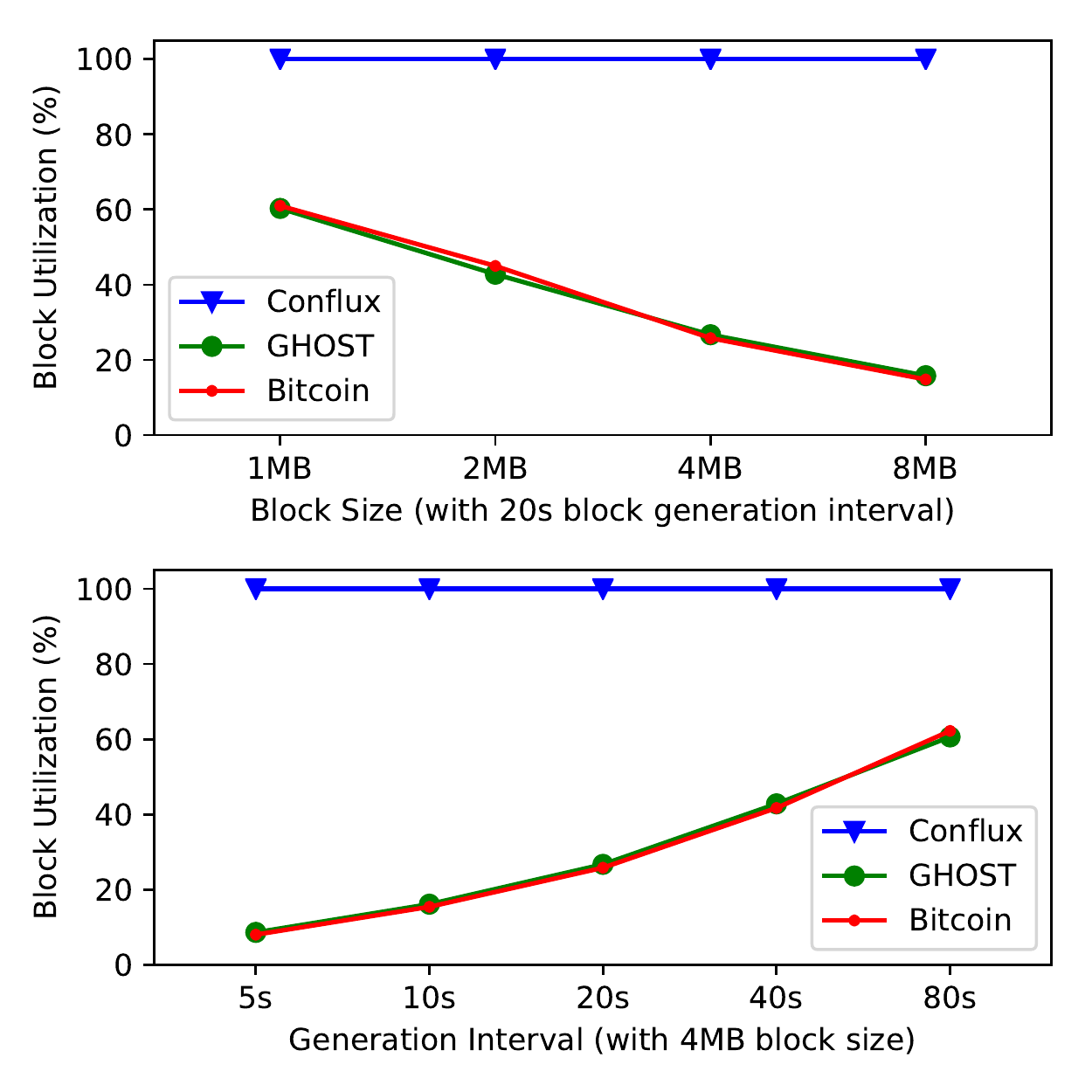}
	\caption{Block utilization ratio.}
    \label{fig:blockutil}%
\end{figure}

To evaluate the throughput improvement, we run {\name}, Bitcoin,
and GHOST with 10k full nodes (i.e., 400 VMs) under the following
two configurations: 
1) increasing the block size limit from 1MB to 8MB with fixed 
block generation rate at 20 seconds per block;
2) decreasing the block generation rate from 5 seconds per block to 
80 seconds per block with fixed block size limit at 4MB.
We run each protocol for two hours for each configuration.

Figure~\ref{fig:blockutil} presents the results,
where X-axis corresponds to different
configuration settings and Y-axis tells the 
correspondent block utilization ratio. 
For Bitcoin and GHOST, this ratio
corresponds to the portion of blocks in the selected chain. For {\name}, this
number is always 1 because all blocks will be eventually included.
Note that the consensus protocol throughput is the
multiplication of three numbers: the block size limit, the block generation
rate, and the block utilization ratio. 

The results show that {\name} achieves a throughput of 2.88GB/h under the block
generation setting 4MB/5s. If we assume the same transaction size as the
real-world Bitcoin network, {\name} could process 3200 transactions per second.
In fact, our results indicate that the throughput of {\name} is only limited
by the processing capability of each individual node, i.e., {\name} can achieve
even higher throughput if we lift the bandwidth limit to 40Mbps (See
Section~\ref{sec:results:scalability}).

Our results also tell that as the block size and
the block generation rate increase, more blocks are generated in
parallel. For Bitcoin and GHOST, this indicates an increasing
number of forks in the resulting block trees. For example, under the block
generation setting 4MB/5s, only 8\% and 8.6\% of blocks are on the agreed chains of
Bitcoin and GHOST, respectively. Blocks in forks will not be included in the
result total order and resources are wasted for generating those blocks. Unlike
Bitcoin and GHOST, {\name} is capable of processing all blocks. {\name}
therefore achieves significantly higher throughputs than Bitcoin and GHOST,
especially when the block size is large or the block generation rate is fast. 


\subsection{Confirmation Time}

\begin{figure}[t]
	\includegraphics[width=\linewidth]{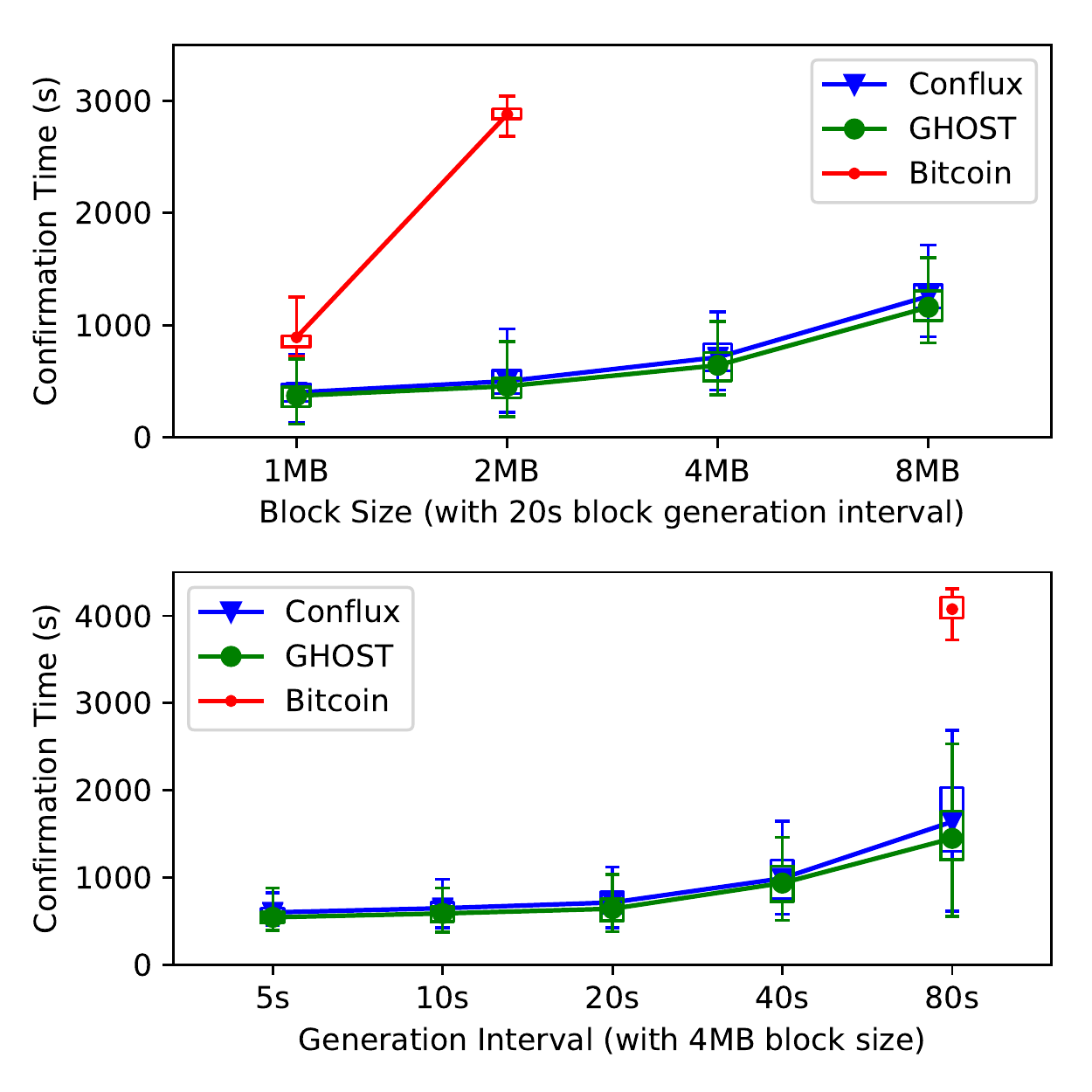}
	\caption{Confirmation time.}
	\label{fig:confirmation}%
\end{figure}

Figure~\ref{fig:confirmation} presents the average confirmation time of {\name},
Bitcoin, and GHOST under different configurations 
as in the same experiments above.
Confirmation time is the
time duration that a user has to wait to obtain a high confidence that the
total order of a block will not change (i.e., the prefix of the total order
before this block does not change). In this setting, 
the user confirms a block if the attacker has
less than 0.01\% chance to revert its transaction, assuming the attacker
controls less than 20\% of the network block generation power (i.e., $q <
0.25$). 
The error bar in Figure~\ref{fig:confirmation} shows the medium, 25\%, 75\%,
minimum, and maximum confirmation time of blocks for each protocol and each
configuration.

Our results show that {\name} can confirm blocks in minutes. When using the
block generation setting of 4MB/5s, {\name} confirms blocks in 10.0min on
average. Under all configurations, users in {\name} wait for a similar
confirmation time as GHOST. This is expected, because the confirmation of
{\name} blocks relies on the confirmation of the corresponding pivot chain
blocks which follows the same GHOST rule. Note that under all settings except
1M/20s, 2M/20s, and 4M/80s, Bitcoin is unable to confirm any block because the
ratio of the longest chain is too small against any attacker with 20\% of the
network block generation power.

The results also show that as the block size increases, blocks take longer to get
confirmed on all three protocols. This is because as explained in
Figure~\ref{fig:blockutil}, using larger blocks will cause more blocks
generated in parallel. Some nodes may temporarily generate blocks that are not
under the last block of the chain (or the pivot chain in {\name}).

Our results further show that increasing the block generation will grow the
chain (or the pivot chain in {\name}) faster and therefore confirm blocks
faster. But this effect diminishes as the block generation rate approaches the
processing capability of individual nodes, because frequent concurrent blocks
and forks will slow down the confirmation. 


\begin{figure}[t]
	\includegraphics[width=\linewidth]{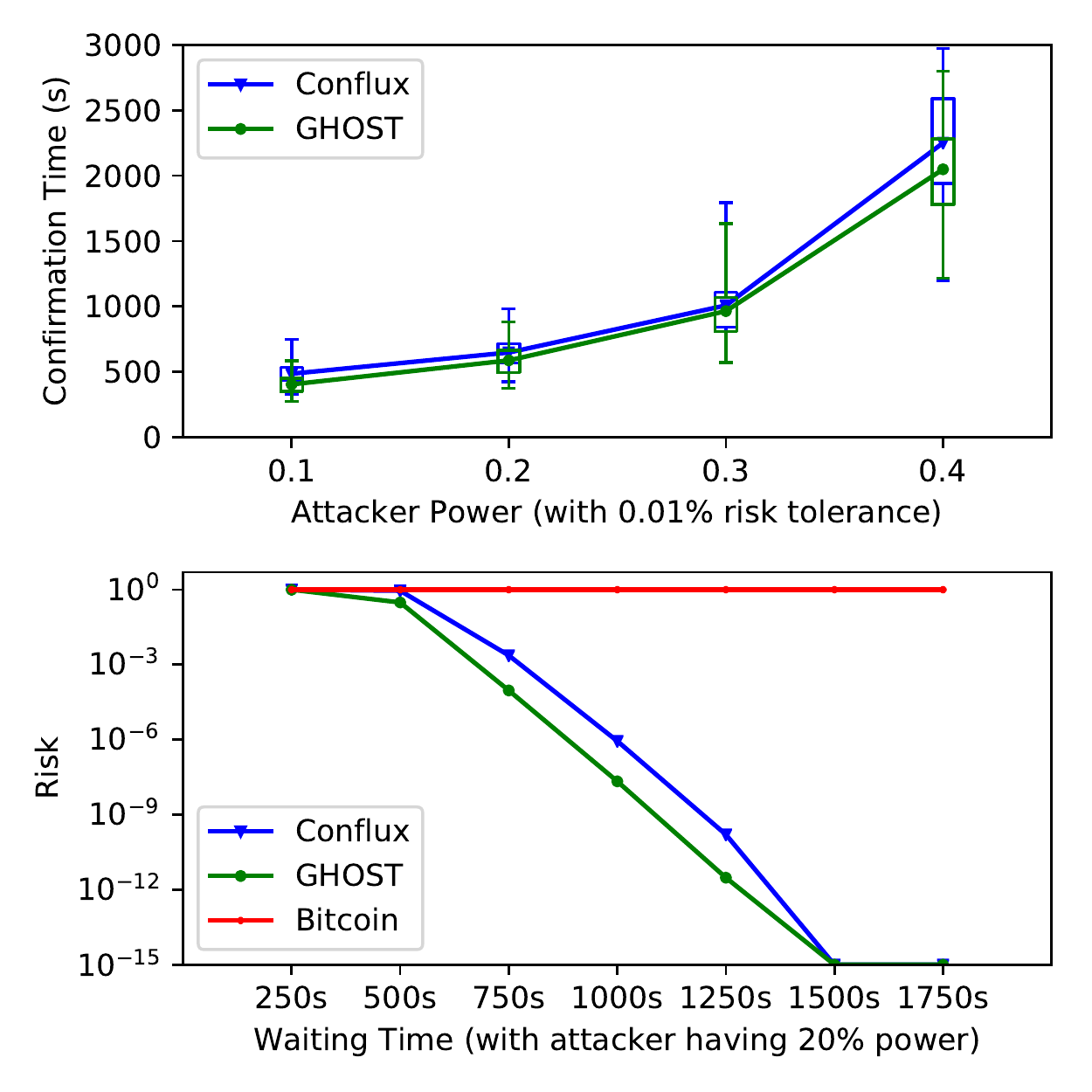}
	\caption{Risk tolerance.}
	\label{fig:attack}%
\end{figure}

\noindent {\textbf{Attacker Capability and Confidence Ratio: }}
The top plot in Figure~\ref{fig:attack} shows, under the block generation
setting of 4M/10s, how the confirmation time changes for {\name} and GHOST, if
the user assumes different attacker capability. Note that Bitcoin cannot
confirm confirmations in this setting. Our results show that even the user
assumes attackers controlling 30\% of the block generation power, {\name}
can still confirm blocks in a medium of 16.8 minutes with confidence 99.99\%.
As the attacker controls more power, the confirmation time grows exponentially
for both protocols. The bottom plot shows, under
the setting of 4M/10s, as the user waits longer how the confirmation risk
changes for {\name}, Bitcoin, and GHOST. Our results show that the chance of
attackers reverting the total order prefix of a confirmed block drops
exponentially as the user waits longer. 

\subsection{Scalability}
\label{sec:results:scalability}

To study the scalability of {\name}, we run {\name} with the following two
configurations: 1) increasing the number of deployed full
nodes from 2.5k to 20k with block size limit at 4MB and block generation
rate at 10s per block (4MB/10s); 
2) lifting the bandwidth limit to 40Mbps per node with block size limit at 4MB 
and block generation rate at 2.5s per block (4MB/2.5s).

\noindent \textbf{Scalability with Higher Bandwidth Limit:} In our experiments,
we found that the throughput bottleneck of {\name} becomes the processing
capability of each individual node, especially the bandwidth limit. {\name}
cannot run under the block generation setting of 4M/2.5s simply because the
gossip network does not have enough bandwidth to propagate blocks under this
fast rate. We then lifted the bandwidth limit of each node from 20Mbps to
40Mbps and run {\name} on 10k full nodes with the setting of 4M/2.5s again.
{\name} runs successfully and achieves the throughput of 5.76G/h. In
this run, {\name} also confirms blocks in 5.68 minutes on average. 
If we assume the same transaction
size as the real-world Bitcoin network, {\name} would process 6400 transactions
per second.

\begin{figure}
\includegraphics[width=\linewidth]{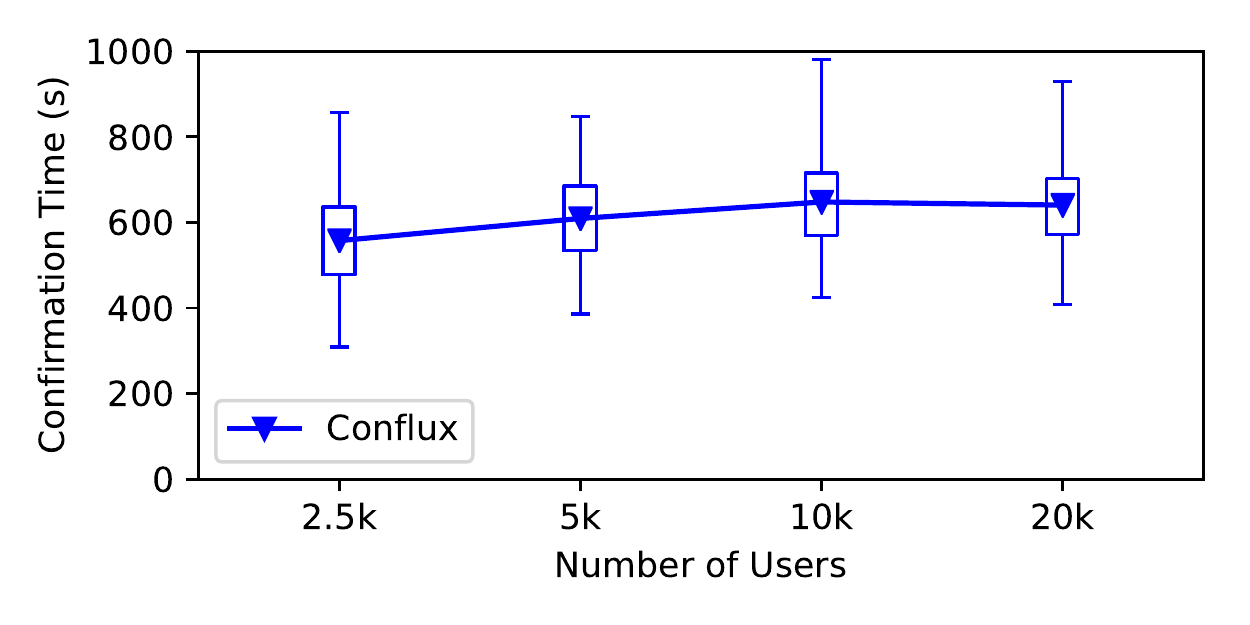}
\caption{Confirmation time of Conflux with 4MB block size, 10s generation interval, and 2.5k to 20k users}
\label{fig:scale}%
\end{figure}

\begin{figure}
\includegraphics[width=\linewidth]{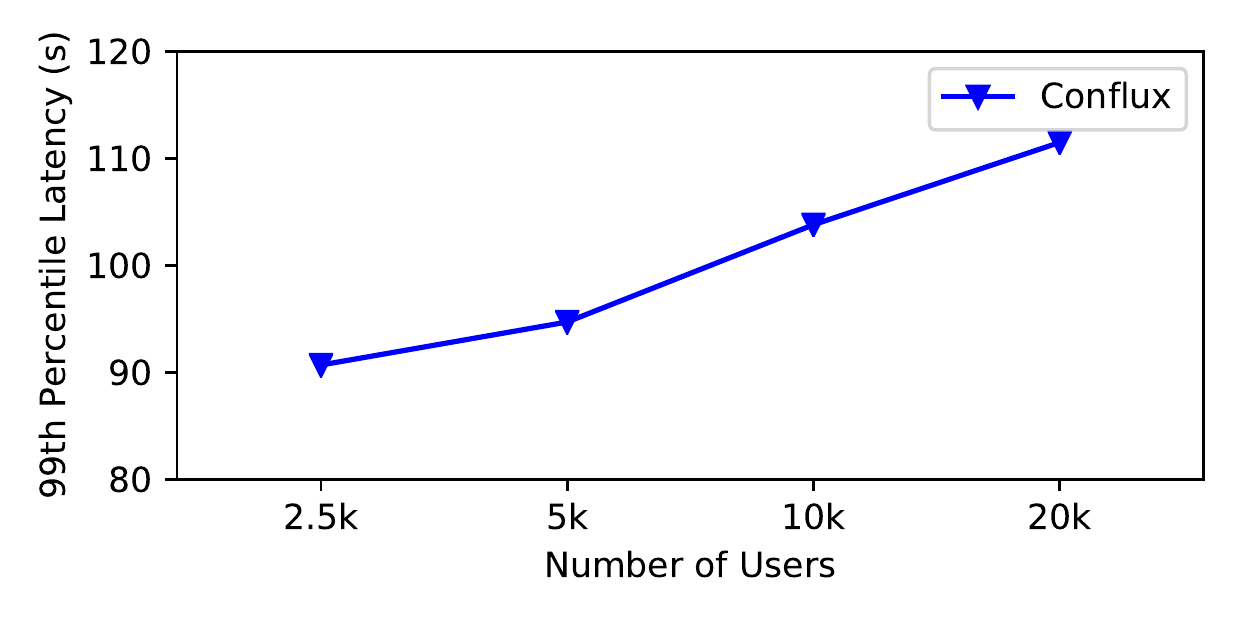}
\caption{Network diameter $d$ of Conflux with 4MB block size, 10s generation interval, and 2.5k to 20k users}
\label{fig:scale-d}%
\end{figure}

\noindent \textbf{Scalability with More Nodes:}
Figure~\ref{fig:scale} shows the confirmation time of {\name} with 2.5k, 5k,
10k, and 20k full nodes under the setting of 4MB/10s. X-axis corresponds to the
number of full nodes while Y-axis tells the confirmation time in seconds.
Figure~\ref{fig:scale-d} shows the network diameter of {\name} with 2.5k, 5k,
10k, and 20k full nodes under the setting of 4M/10s. X-axis corresponds to the
number of full nodes while Y-axis tells the network diameter $d$ (for
propagating 99\% of blocks). Our results show that {\name} scales well to 20k
full nodes and achieves average confirmation time under 10.7 minutes. Note that
the achieved throughput of the 4MB/10s setting is always 720MB/h. Our results
also show that the network diameter grows linearly as the number of full nodes
doubling. The results show that even with 20k users, the increment of the
network diameter is small. Therefore the confirmation time of of a {\name}
transaction is still dominated by waiting enough blocks building on top the
corresponding pivot chain block that processes the transaction. Because we are
using the same 4MB/10s setting, this waiting time stays mostly the same so does
the confirmation time.

\section{Related Work}
\label{sec:related_work}

\noindent \textbf{DAG-based consensus:} People have proposed several consensus
protocols over DAG-based blockchains.  SPECTRE~\cite{SPECTRE} specifies a
non-transitive partial orders for all pairs of blocks in the DAG, while {\name}
provides a total order over all transactions which is critical to support
applications like smart contracts.  Inclusive blockchains~\cite{inclusive}
extends the Nakamoto consensus to DAG and specifies a framework to include
off-chain transactions in a consistent manner. {\name} differs from it in that
{\name} maintains two different kinds of edges between blocks, i.e., parent
edges and reference edges, while inclusive blockchains protocol has only one
kind of edges.  The saperation between parent edges and reference edges
enables us to obtain a pivot chain in the parental tree formed by parent edges only.
We can therefore prove the safety of {\name} directly based on the safety
property of the chain-based GHOST protocol~\cite{GHOST}. In contrast, the main
chain in the inclusive blockchains protocol is a non-extensible path defined in
its DAG, and it is therefore not possible to apply GHOST directly to the
inclusive blockchains protocol.

PHANTOM~\cite{PHANTOM} shares the same aspects of {\name} in terms of
specifying a total order across transactions in DAG-based blockchains. In
PHANTOM, participating nodes first find an approximate $k$-cluster solution for
its local block DAG to prune potentially malicious blocks, then topologically
sorts the remaining blocks to obtain a total order. PHANTOM, however, is
vulnerable to liveness attacks. Attackers with little computation power can
delay the confirmation of transactions indefinitely with high probabilities
even all honest nodes are completely synchronous. See
Appendix~\ref{app:phantom} for the attack.

Besides the aforementioned differences, to our best knowledge {\name} presents
the first empirical evaluation of DAG-based blockchains. Running 10k full nodes
on EC2 where each full node has 40Mbps of bandwidth, {\name} commits 5.76GB of
transactions per hour and confirms them within 4.5-7.4 minutes. There is no
empirical evaluation of other DAG-based protocols and it is therefore unclear
what is the throughput and the confirmation time of these protocols once
implemented and deployed.

\noindent \textbf{Nakamoto consensus:} The Nakamoto consensus protocol
and the GHOST rule specify how the nodes should choose a single
canonical chain when encountering multiple forks~\cite{bitcoin,
  GHOST}. The end result is that all honest nodes converge on the
canonical chain on a high probability. The canonical chain corresponds
to a total-ordered, irreversible log of transactions, where blocks
and transactions that are not on the canonical chains are discarded
and do not contribute to the throughput. 

Instead of choosing one single canonical chain, {\name} assigns a total order
of non-conflicting transactions over the DAG. In {\name}, blocks that are not
on the pivot chains also contribute to the throughput while {\name} still
maintains a total-ordered, irreversible log of transactions for the users,
resulting in significant performance boost compared to Bitcoin and GHOST.

\noindent \textbf{Consortium consensus:} Much research explorers the direction of 
reducing the uses of the expensive Nakamoto consensus in blockchains
to improve their performances. Bitcoin-NG~\cite{BitcoinNG} elects a
leader using the Nakamoto consensus protocol and the leader is
responsible to commit all transactions until the next leader is
elected. Several research work has proposed to combine Nakamoto
consensus with BFT protocols~\cite{ByzCoin, Hybrid}, or to fully
replace Nakamoto consensus with BFT protocols~\cite{Algorand,
  HoneyBadgerBFT, stellar}.

From a practical point of view, all the proposals above run the
alternative consensus protocols within a confined group (one node for
Bitcoin-NG) of nodes, since protocols like BFT only scale up to dozens
of nodes in practices. Therefore one key challenge of these systems
need to address is to choose the confined group in an adversarial
environment like blockchains while maintaining the security
guarantees. For example, the groups can be chosen based on their stakes
of the system~\cite{Algorand} or external hierarchy of
trusts~\cite{stellar}. 

{\name} differs from the above approaches in two ways. First, the total
orders of the transactions is decided by all participants of the
network instead of a confined group. Additionally {\name} is able to
tolerate to half of the network are malicious while the BFT-based
approaches can only tolerate up to one third of malicious
nodes. Second, the above approaches enforce the total order
\emph{eagerly} as the members of the confined group fully 
verify and commit the transactions before moving on to the next
ones. {\name}, however, allows multiple blocks generating in parallel
and finalizes their orders later. The design decision presents an
interesting trade-off between throughput and latency in the
system. For example, {\name} achieves 3.84x throughput compared to
Algorand, but the confirmation time in Algorand is shorter than {\name}.

\noindent \textbf{Fairness:} Recent studies have shown that large miners with
more than 25\% of computational power can capture unproportionally more
rewards, putting smaller miners in disadvantages~\cite{selfishmining,
Fruitchain, Nayak2016}. Although achieving fairness is outside the scope of
this paper, we note that by adapting faster block generation rates and allowing
multiple blocks generated in parallel, {\name} inherently mitigates the
disadvantages of small miners.
It is not possible for a large miner to invalidate blocks generated by small
miners via forking the chain.


\noindent \textbf{Proof-of-Stake:} The original Nakamoto consensus protocol in
Bitcoin requires nodes to solve significant computation puzzles (i.e.,
proof-of-work (PoW)) to vote for consensus. As the PoW scheme demands
a significant amount of resources, alternative schemes such as proof-of-stake
(PoS) has been proposed~\cite{PeerCoin, Ouroboros, Casper}. In PoS based system
the leader is elected based upon the stakes he or she owns in the system. The
leader then is responsible to append new blocks to the blockchain. {\name} is
complementary to the PoS scheme. The consensus algorithm can be adopted by the
PoS-based blockchains as long as the PoS mechanisms can maintain a stable
network block generation rate.

\section{Conclusion}
\label{sec:conclusion}

{\name} is a fast, scalable, and decentralized blockchain platform with proved
safety. It exploits the inherent parallelism among blockchain transactions,
uses a DAG-based approach to defer the total order reconciliation while
providing the externally same interface compared to traditional chain-based
approaches. It provides orders of magnitude throughput improvement, as
validated by the real deployment in Amazon EC2 clusters. {\name} provides a
promising solution to address the performance bottleneck of blockchains and
opens up a wide range of blockchain applications.

\section*{Acknowledgement}
We thank Zhenyu Guo and Haohui Mai for valuable feedbacks on early drafts of
this paper. We thank Guang Yang and Jialin Zhang for the help on the consensus
algorithm.

{\footnotesize \bibliographystyle{acm}
    \bibliography{paper}}
\label{EndPage} 
\appendix

\section{Attack on PHANTOM}
\label{app:phantom}

\subsection{Overview}

PHANTOM~\cite{PHANTOM} is a DAG-based protocol that attempts to achieve
consensus on a total order of blocks. In PHANTOM, participants topologically sorts their local DAG. The algorithm guarantees results consistency among different participants and robustness of accepted transactions. 

The topological sorting algorithm consists of two phases, in the first phase, the algorithm 2-colors all the blocks into blue and red to eliminate the potentially malicious blocks. Given a graph $G$, this phase contains 4 steps:

\begin{itemize}
	\item For each tip (the blocks without decedent) $b\in tips(G)$, let $past(b)$ contains all the ancestors of $b$, coloring the subgraph $past(b)$ recursively.
	\item Let $|BLUE_k(G)|$ denote the number of blue blocks in 2-coloring result of $G$. Find the tip $b_{max}$ which maximizes $|BLUE_k(past(b_{max}))|$.
	\item In graph $G$, color blocks in $past(b_{max})$ according to the result of subgraph $past(b_{max})$, color $b_{max}$ in blue. 
	\item Let $\mathrm{anti}(b)$ denotes the blocks which aren't the ancestors nor the decedents of $b$. For the blocks $b\in \mathrm{anti}(b_{max})$, color it in blue if $\mathrm{anti}(b)$ contains less than $k$ blue blocks in $G$.
\end{itemize}

The score of block $b$ is defined by $|BLUE_k(past(b))|$. A main chain is derived from this step. $b_{max}$ is the chain tip. The highest scoring tip in $past(b_{max})$ is its predecessor in the chain, and so on.

In the second step, participants topological sorts all the blue blocks based on the main chain. In correctness proof of PHANTOM, robustness of topological order is based on robustness of its main chain.


\subsection{Liveness attack}

Here we show an attack for Algorithm 1 in PHANTOM\cite{PHANTOM}. This attack allows attacker to kick out a block from main chain arbitrarily late. We only allow attacker withhold finite blocks and have small block generation capability. For arbitrary large time duration $d$, the attacker is able to kick out one block from main chain which has been received by all the honest miners for time $d$. 

\noindent \textbf{Notation:} In the following, we call the blocks generated by honest nodes \emph{honest blocks} and the blocks generated by attacker \emph{malicious blocks}. For any block $c$, $\mathrm{anti}(c)$ contains all the blocks which are not the ancestors nor the descendants of $c$, which is called \emph{anti-set}. $\mathrm{past}(c)$ contains all the ancestors of $c$.

The attacker chooses an honest block as the start point, which is denoted by $b_1$. The honest blocks which refer $b_1$ as ancestor are denoted by $b_2,b_3,\cdots $ in time sequential. Similarly, we use $a_1,a_2,\cdots$ to denote the malicious blocks. Let set $B$ contain all the honest block, set $A$ contains all the malicious block.

\noindent \textbf{Network Assumption:} Here we use a weaker network assumption.
We do not assume the attacker can control or delay the communication between
honest participants. We instead assume all messages between any two nodes are
delivered immediately, i.e., we assume a fully synchronous network. 

Or we can make this assumption weaker. We only require the messages between attacker and honest nodes are delivered immediately. For the network between honest nodes, we have the following assumption. We assume $|B\cap \mathrm{anti}(b_j)|$ has an upper bound $k'$ for all the $b_j\in B$ with negligible exception. In real world, most mining computation power are in mining pools with good network synchronization. But it also takes several tens of seconds to relay a block to all the participants. So $k'$ may be much smaller than the PHANTOM protocol parameter $k$.

    

We also assume $b_1$ is on the main chain of every honest blocks and $\mathrm{anti}(b_2)\cap B = \emptyset$. 

\noindent \textbf{Parameter Assumption:} Suppose the gap between upper bound $k'$ and PHANTOM protocol parameter $k$ is $k_\Delta=k-k'$. When all the blocks suffer a maximum network delay, the choice of $k$ guarantees that $|\mathrm{anti}(b_j)|\le k$ holds for almost all the honest block $b_j$.  So in a high block generation rate, we can assume that $k$ is large enough. Precisely, 
$$ k_\Delta(k_\Delta -7 )\ge 4 k'.$$

\noindent \textbf{Attack strategy:} We define an positive integer array $\{h_i\}_{i=1}^\infty$ as following:

$$ h_n = \frac{(n-2)(n-1)}{2} +1.$$

For each malicious block $a_i$, it refers all the blocks in $\{b_x |x\in[h_i]\}\cup \{a_y | y\in[i-1]\}$ as its ancestors. Attacker withholds block $a_i$ until block $b_{h_{i-1+k_\Delta}}$ is generated.
 When $b_{h_{i-1+k_\Delta}}$ is generated, attacker makes everyone receive $a_i$ and $b_{h_{i-1+k_\Delta}}$ immediately. 
 
 In this strategy, attacker can start to mine $a_i$ when $b_{h_i}$ and $a_{i-1}$has been generated. If $b_{h_{i-1+k_\Delta}}$ is generated earlier than $a_i$, we say the attacker \emph{fails the liveness attack.} Moreover, if $i\ge 3\Delta - 14$ and $b_{h_{i+1}}$ is generated earlier than block $b_{h_i}$, we will also judge that the attacker fails the liveness attack.

\noindent \textbf{Analysis:} According to the previous strategy, every malicious blocks have a large anti-set and every honest blocks has a small anti-set. 
\begin{lemma}
	For any $b_j\in B$, $|\mathrm{anti}(b_j)\cap A| < k_\Delta$.
	\label{lma:phaa_1}
\end{lemma}
\begin{lemma}
	For any $a_i\in A$, $$|\mathrm{anti}(a_i)\cap B| = \frac{(k_\Delta-1)(k_\Delta+2i-4)}{2}> k + k'.$$
	\label{lma:phaa_2}
\end{lemma}
These properties provides malicious block advantage in calculate blue set.
\begin{lemma}
	According to Algorithm 1 in PHANTOM, for any $b_j\in B$, $Blue_k(past(b_j)) = past(b_j) \cap B$. For any $a_i\in A$, $Blue_k(past(a_i))= past(a_i).$ 
	\label{lma:phanAtt}
\end{lemma}
\begin{proof}
	Without loss of generality, we ignore the common ancestor $past(b_1)$ and regard $b_1$ as genesis block. We prove this lemma by induction.
	
	This lemma holds for $b_1$ trivially since $\mathrm{anti}(b_1)=\emptyset.$
	
	If this lemma holds for all the blocks in $\mathrm{past}(a_i)$, $|Blue_k(past(a_{i-1}))|=h_{i-1}+{i-2}.$ For any honest block $b_j\in past(a_i)$, 
	$|Blue_k(past(b_j))| = |past(b_j)\cap B|=j-1\le h_i-1.$ 
	So we have $|Blue_k(past(a_{i-1}))|>|Blue_k(past(b_j))|$. $a_{i-1}$ is the highest scoring tip of $past(a_i)$.
	
	We also have
	\begin{align*}
		&|\mathrm{anti}(b_j)\cap Blue_k(past(a_{i-1})\cup \{a_{i-1}\}) |\\
		\le &|\mathrm{anti}(b_j)\cap B | +|\mathrm{anti}(b_j)\cap A | \\
		< &k' + k_\Delta.
	\end{align*} 
	
	So all the block in $B\cap past(a_i)$ will be added to $Blue_k(past(a_i))$, which results in $Blue_k(past(a_i))\cap B = past(a_i)\cap B.$ We conclude $Blue_k(past(a_i))=past(a_i)$. 
	$\\$
	
	If this lemma holds for all the blocks in $\mathrm{past}(b_j)$, $\forall j'<j, |Blue_k(past(b_{j'}))|\ge j'-1-k'$. For any attacker block $a_i\in past(b_j)$, we have $j>h_{i-1+k_\Delta}$ and $b_{h_{i-1+k_\Delta}}\in past(b_j)$. 
	
	We have
	\begin{align*}
		&|Blue_k(past(a_i))| \\
		= &i-1+h_i\\
		= &i-1+h_{i-1+k_\Delta} - \frac{(k_\Delta-1)(k_\Delta+2i-4)}{2}\\
		<& h_{i-1+k_\Delta}-1-k' \\
		\le & |Blue_k(past(b_{h_{i-1+k_\Delta}}))| \\ 
		\le& \max_{b_{j'}\in past(b_{j})}|Blue_k(past(b_{j'}))|,
	\end{align*} 
	This inequality shows that the highest scoring tip of $b_j$ is in $B$. We denote it $b_{\bar{j}}$.
	
	For any $a_i\in \mathrm{anti}(G)$,
	\begin{align*}
	&|\mathrm{anti}(a_i)\cap Blue_k(past(b_{\bar{j}})\cup \{b_{\bar{j}}\}) |\\
	\ge &|\mathrm{anti}(a_i)\cap B | -k' \\
	> &k
	\end{align*} 
	No block in $A\cap past(b_j)$ will be add to $Blue_k(past(b_j))$, which means $Blue_k(past(b_j))\cap A = \emptyset$
\end{proof}

From this lemma, we can show that the attacker can kick out block $b_2$ from main chain at any time since the liveness attack has not failed. This means that the honest node shouldn't acknowledge block $b_2$ as an irreversible block because the attacker is able to change the main chain and reorder the blocks. More precisely, we have the following theorem.

\begin{theorem} 
	Starting from the generation time of block $a_{3\Delta-14}$, as long as the attacker has not failed the liveness attack, the attacker is able to kick out block $b_2$ from main chain in all the honest nodes.
\end{theorem}

\begin{proof}
	In the attack strategy, each time the attacker broadcasts $a_i$, it broadcasts $b_{h_{i-1+k_\Delta}}$ at the same time. In the proof of lemma \ref{lma:phanAtt}, we show that $b_{h_{i-1+k_\Delta}}$ has higher score than $a_i$. So for each local graph in an honest nodes, the main chain tip \footnote{The main chain tip is the block on main chain without children.} will be an honest block. We also show that each honest block will also choose an honest block as the highest scoring tip. So the main chain of each honest node must pass block $b_2$.
	
	Since the attacker has not failed the liveness attack, suppose $a_w$ and $b_v$ are the malicious block and honest block with the highest index. So we have $w\ge 3\Delta-14$ and $v<h_{w+2}$. All the honest blocks have score no more than $v-1$ because $Blue_k(past(b_j)) = past(b_j) \cap B\le j$ for all the honest blocks. The score of $a_w$ is $w+h_w-1$. Now we claim that 
	$$ v-1 < h_{w+2} -1 = 2w -2 + h_w  \le w + h_w -1.$$
	
	It implies that if attacker broadcasts $a_w$ and all its ancestors, $a_w$ will be the highest score block in all the honest nodes. The lemma \ref{lma:phanAtt} also shows that the highest scoring tip of a malicious block is also a malicious block. So the main chain of all the honest nodes will pass block $a_1$, which is in $\mathrm{anti}(b_2)$. Attacker kicks out $b_2$ from main chain successfully.
\end{proof}

Now the only remaining problem is that how long can an attacker maintain such an attack. The following theorem gives the probability of an infinite liveness attack.

\begin{theorem}
	Suppose $k_\Delta$ is an even integer. The liveness attack never fails with probability at least $$ \left(1-e^{-cq}\right)^{3k_\Delta-15}\cdot\prod_{i=3k_\Delta-14}^\infty (1-e^{-q(i-1)}),$$
	where $q=\frac{\la}{\lh}, c=1.5k_\Delta-8$. This equation is strictly larger than zero.
\end{theorem}

\begin{proof}
Here we require the attacker to finish the following tasks.

For $a_i$ with $i\le 3k_\Delta-15$, attacker starts to mine $a_i$ when block $b_{(i-1)\cdot c+1}$ is generated. And $a_i$ must be generated before the generation of block $b_{i\cdot c+1}$. 

For $a_i$ with $i\ge 3k_\Delta-14$, attacker starts to mine $a_i$ when block $b_{h_i}$ is generated and $a_i$ must be mined before the generation of $b_{h_{i+1}}$. 

For $i\le 3k_\Delta-15$, we have
\begin{align*}
	h_i &\le (i-1)\cdot c +1 \\
	h_{i-1+k_\Delta} &\ge i\cdot c +1
\end{align*}

For $i\ge 3k_\Delta-14$, we have
\begin{align*}
h_{3k_\Delta-14} &= (3k_\Delta-15)\cdot c +1 \\
h_{i+1} &= h_{i+1}
\end{align*}
	
In this task, it can be shown that the mining time intervals of malicious blocks satisfy the requirements in attack strategy and do not overlap. Completing this task implies applying a successful attack. For $i\le 3k_\Delta-15$, the attacker fails when mining $a_i$ with probability $1-(1-q)^c\le 1-e^{-cq}$. For $i\ge 3k_\Delta-14$, the attacker fails when mining $a_i$ with probability $1-(1-q)^{i-1}\le 1-e^{-q(i-1)}$. The probability of completing this task is 
$$ \left(1-e^{-cq}\right)^{3k_\Delta-15}\cdot\prod_{i=3k_\Delta-14}^\infty (1-e^{-q(i-1)}).$$

Let $n=\left\lceil\frac{-\log\left(1-e^{-q}\right)}{q}\right\rceil+1$, $ e^{-qn}<1-e^{-q}$, we have 
\begin{align*}
	&\prod_{i=3k_\Delta-14}^\infty (1-e^{-q(i-1)}) \\
	\ge& \left(1-\sum_{i=n+1}^\infty e^{-q(i-1)}\right)\prod_{i=3k_\Delta-14}^n (1-e^{-q(i-1)}) \\
	= & \left(1-\frac{e^{-qn}}{1-e^{-q}}\right)\prod_{i=3k_\Delta-14}^n (1-e^{-q(i-1)})\\
	>&0
\end{align*}

So the probability is strictly larger than 0.
\end{proof}
As a result, attacker with $15\%$ computation power can maintain this attack infinitely with probability $98.9\%$ when $k_\Delta \ge 40$.

\end{document}